\documentclass{article}
\usepackage{multicol,caption}
\usepackage{lipsum}
\newenvironment{Figure}
{\par\medskip\noindent}
{\par\medskip}
\usepackage{graphicx}
\usepackage{textcomp}
\usepackage{xcolor}
\usepackage[numbers]{natbib}
\usepackage[ruled,vlined]{algorithm2e}
\usepackage{algpseudocode}
\usepackage{bbm}
\usepackage{mathtools}
\usepackage{caption}
\usepackage{subcaption}
\usepackage{diagbox}
\usepackage{array}
\usepackage{amsmath}
\usepackage{amssymb}
\usepackage{amsfonts}
\usepackage{amsthm}
\usepackage{mathtools}

\usepackage{breqn}
\usepackage{enumitem}
\usepackage{physics}
\usepackage{thm-restate}
\usepackage{url}
\usepackage{tikz-cd}
\theoremstyle{definition}
\newtheorem{definition}{Definition}[section]
\theoremstyle{remark}
\newtheorem{remark}{Remark}
\newtheorem{theorem}{Theorem}[section]

\newcommand{\comment}[1]{}
\allowdisplaybreaks
\usepackage{authblk}
\usepackage[margin=1.2cm]{geometry}
\usepackage{multicol}
\usepackage{blindtext}
\usepackage{epstopdf}
\begin{document}

\title{\textbf{Impact of Sampling on Locally Differentially Private Data Collection}}

\author[1,2]{Sayan Biswas}
\author[3]{Graham Cormode}
\author[4,5]{Carsten Maple}

\affil[1]{{INRIA, France}}
\affil[2]{{LIX, \'{E}cole Polytechnique, France}}
\affil[3]{{Dept. of Computer Science, University of Warwick, UK}}
\affil[4]{{WMG, University of Warwick, UK}}
\affil[5]{{Alan Turing Institute, UK}}
\affil[ ]{{sayan.biswas@inria.fr, g.cormode@warwick.ac.uk, cm@warwick.ac.uk}}
\date{}

\setcounter{Maxaffil}{0}
\renewcommand\Affilfont{\itshape\small}

\maketitle

\begin{abstract}
With the recent bloom of data, there is a huge surge in threats against individuals' private information. Various techniques for optimizing privacy-preserving data analysis are at the focus of research in the recent years. In this paper, we analyse the impact of sampling on the utility of the standard techniques of frequency estimation, which is at the core of large-scale data analysis, of the locally deferentially private data-release under a pure protocol. We study the case in a distributed environment of data sharing where the values are reported by various nodes to the central server, e.g., cross-device Federated Learning. We show that if we introduce some random sampling of the nodes in order to reduce the cost of communication, the standard existing estimators fail to remain unbiased. We propose a new unbiased estimator in the context of sampling each node with certain probability and compute various statistical summaries of the data using it. We propose a way of sampling each node with personalized sampling probabilities as a step to further generalisation, which leads to some interesting open questions in the end. We analyse the accuracy of our proposed estimators on synthetic datasets to gather some insight on the trade-off between communication cost, privacy, and utility.
\end{abstract}

\begin{multicols}{2}

\section{Introduction}
To address the age-old battle between privacy and utility, various optimisation techniques to analyse the data.
There is a massive explosion of data in the recent few years, and with the plethora of data that is being generated everyday, their threats against their privacy is increasing manifold. Hence, the age-old battle between privacy and utility of data has become all the more important catering to the urge to dissect and analyse users' personal data for various kinds of analytics. Differential privacy (DP)~\cite{DworkDP1,DworkDP2} have become the standard for privacy protection in the last few years. To efface the need of a central trusted curator, a local variant of DP called the Local Differential Privacy (LDP)~\cite{DuchiLDP} has been intensively studied of late. With LDP, users get an opportunity to obfuscate their data locally and this noisy data from the end of the users is reported to the central server. The privacy level can be adjusted according to the requirement of the users by , by adding some tuning the privacy parameter $\epsilon$ of the LDP mechanism. 

The idea of LDP aligns well with the modern day distributed machine learning, where the idea is to reduce the dependency on a potentially adversarial central server for carrying out the model training. This gave the rise to the concept of Federated Learning (FL)~\cite{McMahanMRA16FL}, where a local model is trained independently at various nodes and the updates are communicated to the central server to train the main model, aggregating all the local updates. In particular, cross-device Federated Learning \cite{hard2019federated, yang2018applied, chen2019federated, ramaswamy2019federated, leroy2019federated}, the data from the users are used to train a local model on individual devices and the model update is communicated to the central server, making sure that one's personal data never leaves their device. However, in such a setting, often the communication cost is compromised by reporting the updates from every node to estimate the frequency of each value in the domain, and subsequently other statistical summaries of the data, in the central server.

Recently, a substantial focus of FL community focus has been on optimizing sampling techniques~\cite{Abadi2016,dwork2010pan,rizk2020federated,RizkOpt21,WangPureProtocols,Cai2020}. Another branch of recent work has been in the direction of frequency estimation under LDP protocols~\cite{WangPureProtocols,Cormode21}. In this work, we aimed to incorporate the idea of sampling under LDP protocols and analyze its potential impact on standard frequency estimation techniques. 

In this paper, we aim to look at the impact of introducing some sampling techniques on such estimates of LDP data. With millions of users holding data, a useful tool for the service providers is to gather a right number of data points which would be optimal and sufficient for performing various kinds of analysis. In summary, as a main contribution of this work, we point out that the standard estimators fail to stay unbiased when sampling techniques are introduced in a distributed learning framework. Hence, we propose a new unbiased estimator generalising the existing work by Wang et al. in \cite{WangPureProtocols}. We analyse the trade-off between the communication cost and the utility and performance of our estimator under a pure LDP protocol through experiments on synthetic datasets. We illustrate, empirically, the usefulness of sampling by showing that sampling a huge number of users does not drastically improve the quality of the analysis performed after a certain point, therefore, implying the necessity of setting appropriate sampling probabilities to optimize the trade-off between communication cost, privacy, and the quality of the estimators, which contribute massively in the analytics.

\section{Preliminaries }

\begin{definition}[Differential privacy~\cite{DworkDP1,DworkDP2}]
\label{def:dp}
For a certain query, a randomizing mechanism $\mathcal{R}$ provides \emph{$\epsilon$-differential privacy (DP)} if, for all neighbouring\footnote{differing in exactly one place} datasets, $D$ and $D'$, and all $S \subseteq$ Range($\mathcal{R}$), we have
\begin{equation*}
\mathbb{P}[\mathcal{R}(D) \in S] \leq e^{\epsilon}\,\mathbb{P}[\mathcal{R}(D') \in S]
\end{equation*}
\end{definition}

\begin{definition}[Local differential privacy~\cite{DuchiLDP}]
\label{def:ldp}
Let $\mathcal{X}$ and $\mathcal{Y}$ denote the spaces of original and noisy data, respectively. A randomizing mechanism $\mathcal{R}$ provides \emph{$\epsilon$-local differential privacy (LDP)} if, for all $x,\,x'\,\in\,\mathcal{X}$, and all $y\,\in\,\mathcal{Y}$, we have
\begin{equation*}
\mathbb{P}[\mathcal{R}(x)=y] \leq e^{\epsilon}\,\mathbb{P}\left[\mathcal{R}(x')=y\right]
\end{equation*}
\end{definition}

\begin{definition}[Pure LDP protocols]~\cite{WangPureProtocols}
A LDP mechanism $\mathcal{R}$ is \emph{pure} iff there exist $p^*>q^*$ such that for all $v_1$ and $v_2\neq v_1$:
\begin{align}
    \mathcal{P}\left[\mathcal{R}(v_1)\in\{y\colon v_1 \in \operatorname{Support}(y)\}\right]=p^*, \text{ and}\nonumber\\
    \mathbb{P}\left[\mathcal{R}(v_2)\in\{y\colon v_1 \in
    \operatorname{Support}(y)\}\right]=q^*
\end{align}
where for any input $x\in\mathcal{X}$, the set $\{y\in\mathcal{Y}\colon x \in\operatorname{Support}(y)\}$ is the set of outputs in $\mathcal{Y}$ that \emph{support} the input $x$ with a non-zero probability of being observed via the mechanism $\mathcal{R}$, i.e., $\{y\colon x \in\operatorname{Support}(y)\}=\{y\colon\mathbb{P}\left[\mathcal{R}(x)=y\right]\neq 0\}$.
\end{definition}

\begin{definition}[Direct Encoding~\cite{kairouz2016discrete}]
\label{def:DE}
Let $\mathcal{X}$ be a discrete domain of size $d$. Then \emph{direct encoding} (DE) a.k.a. \emph{k-randomized response} ($k$-RR) mechanism, $\mathcal{R}_{\text{DE}}$, is a locally differentially private mechanism that stochastically maps the domain $\mathcal{X}$ onto itself (i.e., $\mathcal{Y}=\mathcal{X}$), given by
\begin{equation*}
    \mathcal{R}_{\text{DE}}(y|x)=
        \begin{cases}
		    p=c\,e^{\epsilon} & \text{, if $x=y$}\\
            q=c, & \text{, otherwise}
		\end{cases}
\end{equation*}
for any $x,\,y\,\in\mathcal{X}$, where $c=\frac{1}{e^{\epsilon}+d-1}$.

\end{definition}

In this work we focus in the setting of DE where it perturbs and fix a discrete domain $\mathcal{X}$ of size $m$ for our analysis, supposing DE perturbs values from $\mathcal{X}$ and to some noisy values in $\mathcal{X}$. Let there are $n\in\mathbb{N}$ nodes, each holding some value from $\mathcal{X}$ obfuscated by DE. Let the Support function for DE be $\text{Support}_{\text{DE}}(i)=\{i\}$, i.e., each obfuscated output value $i\in\mathcal{X}$ supports the input $i\in\mathcal{X}$. 
\begin{remark}\label{rem:DEPure}
    Setting $\text{Support}(i)=\{i\},\,p^*=p$, and $q^*=q$,  DE becomes is pure LDP protocol, shown by Wang et al. in \cite{WangPureProtocols}. 
\end{remark}

Wang et al. \cite{WangPureProtocols} proposed an unbiased frequency estimator, $c_{\text{DE}}(i)$, of the original value $i$ going through a pure LDP protocol as:
\begin{equation}\label{eq:WangUnbiased}
        c(i)=\frac{\sum\limits\limits_j\mathbbm{1}_{\text{Support}(y^j)}(i)-nq^*}{p^*-q^*}
\end{equation}
where $y^j$ denotes the noisy value reported by the $j^{\text{th}}$ node. Thus, using \eqref{eq:WangUnbiased} in the context of DE, for any value $i\in\mathcal{X}$ we obtain:
\begin{equation}\label{eq:DEUnbiased}
    c_{\text{DE}}(i)=\frac{\sum\limits\limits_{j=1}^n\mathbbm{1}_{\{X_j=i\}}-nq}{p-q}
\end{equation}
where $\mathbbm{1}_{E}$ is the indicator function for any event $E$ such that
\begin{equation*}
\mathbbm{1}_E=
    \begin{cases}
		1 & \text{ if $E$ happens}\\
        0, & \text{ otherwise}
	\end{cases}
\end{equation*}

We explore this idea to investigate the behaviour of $c_{\text{DE}}$ if each node is independently sampled to report its value, perturbed with DE, with some probability $\pi$. Let $S$ be the random variable representing the number of nodes which have been reported to the central server. Hence $\mathbb{P}(S>n)=\mathbb{P}(S<0)=0$. Taking the same estimator $c_{\text{DE}}(i)$ in the setup of random sampling of each node with an independent probability of $\pi$, we get:
\begin{align}
    \mathbb{E}\left(c_{\text{DE}}(i)\right)
    =\mathbb{E}\left(\frac{\sum\limits\limits_{j=1}^S\mathbbm{1}_{\text{Support}(y^j)}(i)-nq^*}{p^*-q^*}\right)\nonumber\\
    =\mathbb{E}\left(\frac{\sum\limits\limits_{j=1}^S\mathbbm{1}_{\{X_j=i\}}-nq}{p-q}\right)\nonumber\\
    =\frac{\mathbb{E}(S)\mathbb{E}(\mathbbm{1}_{\{X_{j}=i\}}-nq)}{p-q} \text{ [Wald's Equation~\cite{Wald}]}\nonumber\\
    =\frac{n\pi(f_ip+(1-f_i)q)-nq}{p-q}\nonumber\\
    =nf_i\pi-\frac{nq(1-\pi)}{p-q}\label{eq:WangBiased1}
\end{align}

We see that putting $\pi=1$ in \eqref{eq:WangBiased1}, implying every node is sampled in each round, gives us the same result as in \cite{WangPureProtocols}. 
\begin{restatable}{theorem}{th:WangEstBiased}\label{th:WangEstBiased}
    If we introduce some sampling probability $\pi<1$ for each node, $c_{\text{DE}}$ becomes a biased frequency estimator.
\end{restatable}
\begin{proof}
    We recall that $1\geq p\geq q \geq 0$ by the definition of pure LDP protocols, and $0\leq\pi\leq 1$. Therefore, $\frac{nq(1-\pi)}{p-q}\geq 0$, and hence, $\mathbb{E}(c_{\text{DE}}(i)) \leq nf_i\pi \leq nf_i$ and equality is attained iff $\pi=1$.
\end{proof}

\section{Unbiased frequency estimation}
Motivated from Theorem \ref{th:WangEstBiased}, we proceed to device an unbiased estimator for DE, $g_{\text{DE}}$, incorporating the random sampling aspect, defined as follows:
\begin{equation}\label{eq:unbiased1}
    g_{\text{DE}}(i)=\frac{c_{\text{DE}}(i)}{\pi}+\frac{nq(1-\pi)}{(p-q)\pi}  
\end{equation}

\begin{restatable}{theorem}{th:unbiasedmean}\label{th:unbiasedmean}
If each node has an independent sampling probability of $\pi$, $g_{\text{DE}}$ is an unbiased estimator of the frequencies of the values in $\mathcal{X}$ observed under DE. 
\end{restatable}

\begin{proof}
Immediate from \eqref{eq:WangBiased1} in Theorem \ref{th:WangEstBiased} and using the linearity of expectation. 
\end{proof}

For the simplicity of notation, let $f_i$ be the random variable representing the fraction of times the value $i\in\mathcal{X}$ is reported to the central server. In \cite{WangPureProtocols} Wang et al. define the \emph{approximate variance} of any random variable which is a function of $f_i$, say $RV(f_i)$, as  $\operatorname{Var}^*(RV(f_i))=\lim\limits_{f_i\rightarrow 0}\operatorname{Var}(RV(f_i))$.

\begin{restatable}{theorem}{th:unbiasedvar}\label{th:unbiasedvar}
In the event of independently sampling the nodes with some probability $\pi$, the approximate variance of $c_{\text{DE}}$ is given by: $$\operatorname{Var}^*(g_{\text{DE}}(i))=\frac{\operatorname{Var}^*(c_{\text{DE}}(i))}{\pi^2}=\frac{n(q-q^2\pi)}{(p-q)^2\pi}$$
\end{restatable}

\begin{proof}
We start by deriving the actual variance of $g_{\text{DE}}$.
\begin{align}
\operatorname{Var}\left(c_{\text{DE}}(i)\right)
=\operatorname{Var}\left(\frac{\sum\limits\limits_{j=1}^S\mathbbm{1}_{\{X_j=i\}}-nq}{p-q}\right) \nonumber \\
\text{[$S$ is the r.v. representing the number of nodes sampled]}\nonumber\\
=\frac{\operatorname{Var}\left(\sum\limits\limits_{j=1}^S\mathbbm{1}_{\{X_j=i\}}-nq\right)}{(p-q)^2} \nonumber \\
=\frac{\operatorname{Var}\left(\sum\limits\limits_{j=1}^S\mathbbm{1}_{\{X_j=i\}}\right)}{(p-q)^2} \nonumber \\
=\frac{\mathbb{E}(S)\operatorname{Var}(\mathbbm{1}_{\{X_1=i\}})+\mathbb{E}((\mathbbm{1}_{\{X_1=i\}})^2\operatorname{Var}(S)}{(p-q)^2}\nonumber \\
\text{[Random sums of RVs~\cite{Wald}]}\nonumber\\
=\frac{n\pi f_i p(1-p) +n\pi (1-f_i)q(1-q)}{(p-q)^2} \nonumber \\
+\frac{(f_ip+(1-f_i)q)^2(n\pi(1-\pi))}{(p-q)^2} \nonumber\\
\therefore \operatorname{Var}^*(c_{\text{DE}}(i))= n\pi\frac{q(1-q)+q^2(1-\pi)}{(p-q)^2} \nonumber \\
= n\pi\frac{q-q^2\pi}{(p-q)^2}\nonumber
\end{align} 

Now we observe that $\operatorname{Var}(g_{\text{DE}}(i))=\frac{\operatorname{Var}(c_{\text{DE}}(i))}{\pi^2}$, by definition of $g_{\text{DE}}$. Therefore, 
\begin{align}
\operatorname{Var}^*(g_{\text{DE}}(i))=\frac{\operatorname{Var}^*(c_{\text{DE}}(i))}{\pi^2}=\frac{n(q-q^2\pi)}{(p-q)^2\pi}\nonumber
\end{align}
\end{proof}

Observe $\operatorname{Var}^*(g_{\text{DE}}(i)) \geq \operatorname{Var}^*c_{\text{DE}}(i))$, with equality iff $\pi=1$, as we would expect since we are introducing more randomness and less information in $g_{\text{DE}}(i)$ compared to $c_{\text{DE}}(i)$ by engendering random sampling of each node. 

\begin{definition}[Normalized variance]
    The \emph{normalised variance} of any random variable $X$ is defined as  
    $$\operatorname{Var}_{\text{norm}}(X)=\frac{\operatorname{Var}(X)}{\mathbb{E}(X)}$$
\end{definition}

Normalized variance can be useful when comparing two random variables with different means, in order to account for larger variance for larger means.

\begin{theorem}
\label{normvarg}
 $\operatorname{Var}^*_{\text{norm}}(g_{\text{DE}}(i))=\mathcal{O}\left(\frac{1}{\pi^3 n}\right)$
\end{theorem}
\begin{proof}
\begin{align}
\operatorname{Var}_{\text{norm}}\left(g_{\text{DE}}(i)\right)
= \operatorname{Var}\left(\frac{g_{\text{DE}}(i)}{\mathbb{E}(S)}\right) \nonumber \\
=\operatorname{Var}\left(\frac{g_{\text{DE}}(i)}{n\pi}\right)
=\frac{\operatorname{Var}(g_{\text{DE}}(i))}{n^2\pi^2}\nonumber\\
\implies \operatorname{Var}^*\left(\frac{g_{\text{DE}}(i)}{n\pi}\right)=\frac{\operatorname{Var}^*(g_{\text{DE}}(i))}{n^2\pi^2}\nonumber\\
=\frac{n(q-q^2\pi)}{(p-q)^2n^2\pi^3} \text{ [Th.\ref{th:unbiasedvar}]}
=\frac{q-q^2\pi}{(p-q)^2\pi^3n}
= \mathcal{O}\left(\frac{1}{\pi^3 n}\right) \nonumber
\end{align}
\end{proof}

We note that for small value of $\pi$, the normalized variance of the estimator $g_{\text{DE}}$ would blow up as it is of the order $\frac{1}{\pi^3 n}$. But this is not unexpected, as with a low sampling probability, it is more likely that we would give rise to fewer nodes that are actually sampled to report their values, giving rise to less information for the central server, which should result in a greater variance. We acknowledge a trade-off between the bias of an estimator and its increasing variance. In particular, we see that without compensating for the bias of $c_{\text{DE}}$ to obtain $g_{\text{DE}}$ by scaling it with $\frac{1}{\pi}$ and adding up $\frac{nq(1-\pi)}{\pi(p-q)}$, for a small sampling probability $\pi$, we would have the bias which will grow up to be a tremendously low a quantity, always giving a massively conservative and negative estimate for the value of $i$ as observed, especially if the number of nodes involved ($n$) is huge (e.g. in millions), which is often the case in federated learning. Precisely, observe from \eqref{eq:WangBiased1} that as $\lim\limits_{\pi\to0}c_{\text{DE}}=\frac{nq}{p-q}$, implying that we would be getting a constant and negative estimate for every $i\in\mathcal{X}$, which would make the analysis involving the frequencies rather absurd.

Now we look to improve upon the proposed unbiased frequency estimator $g_{\text{DE}}$. Let $S$ be the random variable representing the number of nodes sampled in a round if each node is independently sampled with probablity $\pi$. We proceed to define an improved frequency estimator of the elements of $\mathcal{X}$ under DE through a very natural approach of replacing $n$ by $S$ in the definition of $c_{\text{DE}}$.

Let $\hat c_{\text{DE}}(i)=\frac{\sum\limits\limits_{j=1}^S\mathbbm{1}_{\{X_j=i\}}-Sq}{\pi(p-q)}$. In order to use $\hat{c}_{DE}$ as the frequency estimator for any element $i\in\mathcal{X}$, it is crucial to probe if it has any bias. 
\begin{restatable}{theorem}{th:betterunbiased}\label{th:betterunbiased}
$\hat c_{\text{DE}}$ is an unbiased estimator of the frequencies of the elements of $\mathcal{X}$ being perturbed via DE which are reported by the nodes which are sampled independently.
\end{restatable}

\begin{proof}
\begin{align}
\mathbb{E}\left(\frac{\sum\limits\limits_{j=1}^S\mathbbm{1}_{\{X_j=i\}}-Sq}{\pi(p-q)}\right) \nonumber \\
= \frac{\mathbb{E}(S)\mathbb{E}(\mathbbm{1}_{\{X_{j}=i\}})-\mathbb{E}(Sq)}{\pi(p-q)}\text{ [Wald's Equation~\cite{Wald}]}\nonumber \\
=\frac{n\pi(f_ip+(1-f_i)q)-n\pi q}{\pi(p-q)} 
= nf_i\nonumber
\end{align}
\end{proof}

\begin{theorem}
$\operatorname{Var}(\hat c_{\text{DE}}(i)) \geq \operatorname{Var}(g_{\text{DE}}(i))$, i.e., $g_{\text{DE}}$ gives a better (more confident) estimate for the frequencies than $\hat c_{\text{DE}}$, which is a naive and immediate extension from $c_{\text{DE}}$. 
\end{theorem}

\begin{proof}
\begin{align}
\operatorname{Var}\left(\hat c_{\text{DE}}(i)\right)      \nonumber \\
=\operatorname{Var}\left(\frac{\sum\limits\limits_{j=1}^S\mathbbm{1}_{\{X_j=i\}}-Sq}{\pi(p-q)}\right) \nonumber \\
= \frac{\operatorname{Var}\left(\sum\limits\limits_{j=1}^S\mathbbm{1}_{\{X_j=i\}}\right)+\operatorname{Var}(S)q^2}{\pi^2(p-q)^2} \nonumber \\
= \operatorname{Var}(g_{\text{DE}}(i))+\frac{\operatorname{Var}(S)q^2}{\pi^2(p-q)^2}\nonumber \\
\end{align}
It follows immediately that $\operatorname{Var}(g_{\text{DE}}(i))+\frac{\operatorname{Var}(S)q^2}{\pi^2(p-q)^2} \geq \operatorname{Var}(g_{\text{DE}}(i)$ as $\frac{\operatorname{Var}(S)q^2}{\pi^2(p-q)^2}\geq 0$.
\end{proof}

\begin{restatable}{theorem}{equival}\label{th:equival}
For every $i\in\mathcal{X}$, we have $0\leq g_{\text{DE}}(i)\leq n$ iff $0\leq c_{\text{DE}}(i)\leq n$ on an average, i.e., ensuring our proposed frequency estimate evaluating a reasonable frequency for any $i\in\mathcal{X}$ is equivalent to that of the estimate proposed by Wang et al.
\end{restatable}

\begin{proof}
We proceed to show this in two parts: 
\begin{enumerate}
    \item[(i)] $0 \leq c_{\text{DE}}(i) \Leftrightarrow 0 \leq g_{\text{DE}}(i)$ on an average
    \item[(ii)] $n \geq c_{\text{DE}}(i) \Leftrightarrow n \geq g_{\text{DE}}(i)$ on an average
\end{enumerate}

Proceeding with (i), we obtain:
\begin{align}
c_{\text{DE}}(i) \geq 0
\Leftrightarrow \frac{\sum\limits\limits_{j=1}^n\mathbbm{1}_{\{X_j=i\}}-nq}{p-q} \geq 0\nonumber\\
\Leftrightarrow \sum\limits\limits_{j=1}^n\mathbbm{1}_{\{X_j=i\}}-nq \geq 0 \text{ [$p\geq q$ for pure LDP]}\nonumber\\
\Leftrightarrow \sum\limits\limits_{j=1}^n\mathbbm{1}_{\{X_j=i\}} \geq nq
\Leftrightarrow \mathbb{E}\left(\sum\limits\limits_{j=1}^n\mathbbm{1}_{\{X_j=i\}}\right) \geq nq\nonumber\\
\Leftrightarrow n(f_ip+(1-f_i)q) \geq nq 
\Leftrightarrow p \geq q
\end{align}
That's the trivial condition assumed to make DE a pure LDP protocol.

Now focussing on $g_{\text{DE}}$, we get:
\begin{align}
g_{\text{DE}}(i) \geq 0
\Leftrightarrow\frac{\sum\limits\limits_{j=1}^S\mathbbm{1}_{\{X_j=i\}}-nq}{\pi(p-q)}+\frac{nq(1-\pi)}{(p-q)\pi} \geq 0\nonumber\\
\Leftrightarrow\sum\limits\limits_{j=1}^S\mathbbm{1}_{\{X_j=i\}}-nq\pi \geq 0 \text{ [$p\geq q$ for pure LDP]}\nonumber\\
\Leftrightarrow\sum\limits\limits_{j=1}^S\mathbbm{1}_{\{X_j=i\}} \geq nq\pi
\end{align}

Taking the expectation of both sides: 
\begin{align}
\Leftrightarrow\mathbb{E}\left(\sum\limits\limits_{j=1}^S\mathbbm{1}_{\{X_j=i\}}\right) \geq nq\pi \nonumber\\
\Leftrightarrow n\pi(f_ip+(1-f_i)q) \geq nq\pi \nonumber\\
\Leftrightarrow p \geq q
\end{align}

Establishing (i), now we shift to prove (ii):
\begin{align}
c_{\text{DE}}(i) \leq n
\Leftrightarrow\frac{\sum\limits_{j=1}^n\mathbbm{1}_{\{X_j=i\}}-nq}{p-q} \leq n\nonumber\\
\Leftrightarrow\sum\limits_{j=1}^n\mathbbm{1}_{\{X_j=i\}}-nq \leq n(p-q) \nonumber\\
\Leftrightarrow\sum\limits_{j=1}^n\mathbbm{1}_{\{X_j=i\}} \leq np\nonumber\\
\Leftrightarrow\mathbb{E}\left(\sum\limits_{j=1}^n\mathbbm{1}_{\{X_j=i\}}\right) \leq np\nonumber\\
\Leftrightarrow n(f_ip+(1-f_i)q) \leq np \nonumber\\
\Leftrightarrow q \leq p
\end{align}
But $q \leq p$ is the trivial condition assumed to make direct encoding a pure LDP protocol.

\begin{align}
g_{\text{DE}}(i) \leq n
\Leftrightarrow\frac{\sum\limits_{j=1}^S\mathbbm{1}_{\{X_j=i\}}-nq}{\pi(p-q)}+\frac{nq(1-\pi)}{(p-q)\pi} \leq n\nonumber\\
\Leftrightarrow\sum\limits_{j=1}^S\mathbbm{1}_{\{X_j=i\}}-nq\pi \leq n(p-q)\pi \text{ [$p>q$ for pure]}\nonumber\\
\Leftrightarrow\sum\limits_{j=1}^S\mathbbm{1}_{\{X_j=i\}} \leq np\pi\nonumber\\
\Leftrightarrow\mathbb{E}\left(\sum\limits_{j=1}^S\mathbbm{1}_{\{X_j=i\}}\right) \leq np\pi \nonumber\\
\Leftrightarrow n\pi(f_ip+(1-f_i)q) \leq np\pi
\Leftrightarrow q \leq p
\end{align}
\end{proof}

\section{Experimental results}
We performed experiments on synthetic datasets to evaluate and visualize the performance of our estimator and observed that, indeed, as we increase $\pi$, the estimation by $g_{\text{DE}}(i)$ approximates the original distribution better. We considered 50,000 data points sampled from a domain $\mathcal{X}$ of size 100, following the distributions Binomial$(100,0.5)$ and Binomial$(50,0.6)$+Binomial$(50,0.4)$. We considered two extremes of the sampling probabilities for each node by setting $\pi=0.1$ and $\pi=0.9$. Figures \ref{binomest} \& \ref{bimodest} illustrate the performance of our estimator in these two settings for the two different datasets.

We computed the total variation (TV) distance between the original distribution and $\hat  c_{\text{DE}}(i)$ for the synthetically generated dataset sampled from a $Bin(100,0.5)$ distribution and illustrated the results in Figure \ref{TV dist}, along with communication cost for sampling probabilities ranging from $\pi=0.1$ to $\pi=0.9$. We can see a clear trade-off between the communication cost and the TV distance.

\begin{Figure}
\begin{center} 
\scalebox{0.9}{
\includegraphics[scale=0.36]{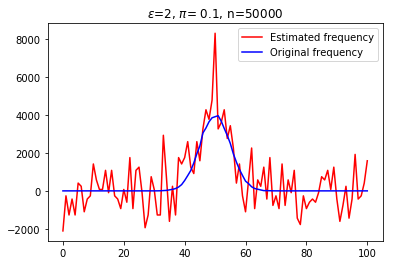}
\includegraphics[scale=0.36]{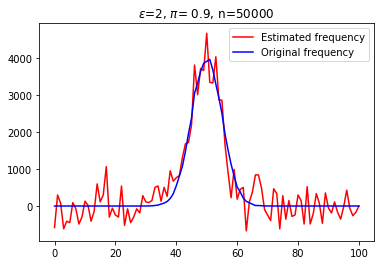}
}
\end{center}
\captionof{figure}{Data sampled from $Bin(100,0.5)$}
\label{binomest}
\end{Figure}
\begin{Figure}
\begin{center} 
\scalebox{0.9}{
\includegraphics[scale=0.36]{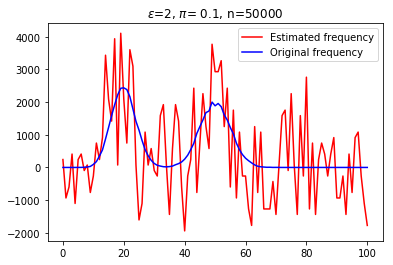}
\includegraphics[scale=0.36]{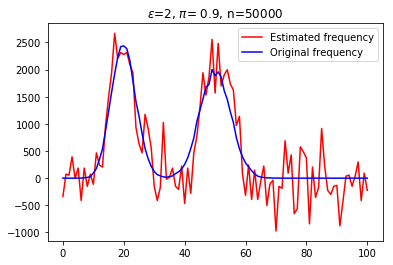}
}
\end{center}
\captionof{figure}{Data sampled from $Bin(50,0.6)+Bin(50,0.4)$}
\label{bimodest}
\end{Figure}

\begin{Figure}
\begin{center} 
\scalebox{0.9}{
\includegraphics[scale=0.36]{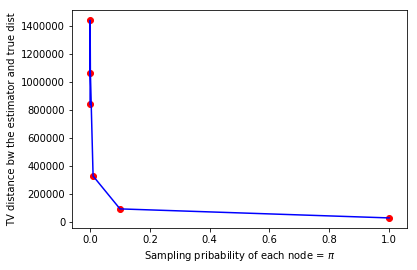}
\includegraphics[scale=0.36]{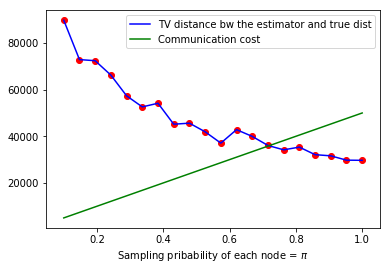}
}
\captionof{figure}{Total variation distance between our proposed estimator and distribution of the original data, and communication cost = $\mathcal{O}(n\pi)$ varying with different sampling probabilities}
\label{TV dist}
\end{center}
\end{Figure}

\section{Generalized sampling probabilities}
In all the previous results, we assumed that the values from each node is sampled independently with the same probability $\pi$. Now we enable us with the flexibility not to require the sampling probability of each node to be the same, opening doors to a lot of interesting paths of research ahead. We explore the setting where the $j^{\text{th}}$ node is sampled independently with probability $\pi_{j}$  for every node $j \in \{1,2,\ldots,n\}$. Note that if we have $\pi = \pi_1=\pi_2=\ldots\pi_n$, we are left with the sampling environment that we addressed previously. 

Let $S$ be the random variable representing the total number of nodes sampled under this flexible setting of having personalized sampling probabilities. We proceed to derive an unbiased frequency estimator in such a generalized case. 

\begin{restatable}{theorem}{th:generalest}
    Let $h(i)=\frac{\sum\limits_{j=1}^S\mathbbm{1}_{\{X_j=i\}}-nq}{p-q}$, where $X_j$ is the random variable denoting the value reported by the $j^{\text{th}}$ node. Then, setting $\mathcal{T}(i)$ as
    $\frac{nh(i)}{\sum\limits_{j=1}^n\pi_j}+\frac{nq\left(n-\sum\limits_{j=1}^n\pi_j\right)}{\sum\limits_{j=1}^n\pi_j(p-q)}$, it
    becomes an unbiased frequency estimator of every value $i \in \mathcal{X}$ with $$\operatorname{Var}^*(\mathcal{T}(i))=\frac{n^2\sum\limits_{j=1}^n(q\pi_j(1-q\pi_j))}{(\sum\limits_{j=1}^n\pi_j)^2(p-q)^2}$$
\end{restatable}

\begin{remark}\label{rem:relatingvar}
Putting $\pi=\pi_1\ldots=\pi_n$ reduces $\operatorname{Var}^*(\mathcal{T}(i))$ to  $\operatorname{Var}^*(g_{\text{DE}}(i))$ and further, putting  $\pi_1=\ldots\pi_n=1$  reduces $\operatorname{Var}^*(\mathcal{T}(i))$ to  $\operatorname{Var}^*(c_{\text{DE}}(i))$ as in \cite{WangPureProtocols}, as expected.
\end{remark}

\begin{proof}
First we aim to show that $\mathcal{T}(i)$ is an unbiased estimator for any $i\in\mathcal{X}$.

\small
\begin{align}
\mathbb{E}(h(i))
= \mathbb{E}\left(\frac{\sum\limits_{j=1}^S\mathbbm{1}_{\{X_j=i\}}-nq}{p-q}\right) \nonumber \\
= \frac{\mathbb{E}\left(\sum\limits_{j=1}^n\mathbbm{1}_{\{X_{j} \text{ is sampled}\}}\mathbbm{1}_{\{X_{j}=i\}}-nq\right)}{p-q}  \nonumber \\
\text{[As sampling \& privatization are independent]}\nonumber\\
= \frac{\sum\limits_{j=1}^n\mathbb{E}\left(\mathbbm{1}_{\{X_{j} \text{ is sampled}\}}\right)\mathbb{E}\left(\mathbbm{1}_{\{X_{j}=i\}}\right)-nq}{p-q}  \nonumber \\
= \frac{\sum\limits_{j=1}^n\mathbb{P}\left(\mathbbm{1}_{\{X_{j} \text{ is sampled}\}}\right)\mathbb{P}\left(\mathbbm{1}_{\{X_{j}=i\}}\right)-nq}{p-q} \nonumber \\
= \frac{\sum\limits_{j=1}^n\pi_j(f_ip+(1-f_i)q) - nq}{p-q}
= \sum\limits_{j=1}^n\pi_jf_i-\frac{q(n-\sum\limits_{j=1}^n\pi_j)}{p-q}\nonumber
\end{align}
\normalsize

Note that  $\mathbb{E}(S)=\mathbb{E}\left(\mathbbm{1}_{\{x_j \text{ is sampled}\}}\right)=\sum\limits_{j=1}^n\pi_j \leq n$. Therefore, $\frac{q(n-\sum\limits_{j=1}^n\pi_j)}{p-q} \geq 0$. Hence, we define $\mathcal{T}(i)=\frac{nh(i)}{\sum\limits_{j=1}^n\pi_j}+\frac{nq(n-\sum\limits_{j=1}^n\pi_j)}{\sum\limits_{j=1}^n\pi_j(p-q)}$ as the frequency estimate of the true value $i$. Because of linearity of expectation, we get $\mathbb{E}(\mathcal{T}(i))=nf_i$, giving us an unbiased estimator for the general case where each node can have a different probability of being sampled. Putting $\pi = \pi_1=\pi_2=\ldots\pi_n$ reduces $\mathcal{T}(i)$ to $g_{\text{DE}}(i)$ which is what we would expect.  

Now we focus on computing $\operatorname{Var}^*\left(\mathcal{T})(i)\right)$ by first evaluating the actual variance of $\mathcal{T}(i)$. We obtain:

\small
\begin{align}
\operatorname{Var}(\mathcal{T}(i))
= \operatorname{Var}\left(\frac{nh(i)}{\sum\limits_{j=1}^n\pi_j}+\frac{q\left(n-\sum\limits_{j=1}^n\pi_j\right)}{\sum\limits_{j=1}^n\pi_j(p-q)}\right)
= \frac{n^2\operatorname{Var}(h(i))}{(\sum\limits_{j=1}^n\pi_j)^2}\nonumber
\end{align}

\begin{align}
= \frac{n^2}{(\sum\limits_{j=1}^n\pi_j)^2}\operatorname{Var}\left(\frac{\sum\limits_{j=1}^S\mathbbm{1}_{\{X_j=i\}}-nq}{p-q}\right)
= \frac{n^2\operatorname{Var}\left(\sum\limits_{j=1}^S\mathbbm{1}_{\{X_j=i\}}\right)}{(\sum\limits_{j=1}^n\pi_j)^2(p-q)^2} \nonumber \\
= \frac{n^2\operatorname{Var}\left(\sum\limits_{j=1}^n\mathbbm{1}_{\{X_j\text{ is sampled}\}}\mathbbm{1}_{\{X_j=i\}}\right)}{(\sum\limits_{j=1}^n\pi_j)^2(p-q)^2} \nonumber \\
= \frac{n^2}{(\sum\limits_{j=1}^n\pi_j)^2(p-q)^2}\left(\sum\limits_{j=1}^n(\operatorname{Var}\left(\mathbbm{1}_{\{X_j\text{ is sampled}\}}\right)\operatorname{Var}\left(\mathbbm{1}_{\{X_j=i\}}\right) \right.\nonumber\\
\left.+\operatorname{Var}\left(\mathbbm{1}_{\{X_j\text{ is sampled}\}}\right)\mathbb{E}\left(\mathbbm{1}_{\{X_j=i\}}\right)^2\right.\nonumber\\
\left.+\mathbb{E}\left(\mathbbm{1}_{\{X_j\text{ is sampled}\}}\right)^2\operatorname{Var}\left(\mathbbm{1}_{\{X_j=i\}}\right)\right)\nonumber \\
= \frac{n^2\sum\limits_{j=1}^n(\pi(1-\pi)((f_ip(1-p)+(1-f_i)q(1-q))}{(\sum\limits_{j=1}^n\pi_j)^2(p-q)^2} \nonumber \\
+  \frac{(f_ip+(1-f_i)q)^2)}{(\sum\limits_{j=1}^n\pi_j)^2(p-q)^2}
+ \pi^2(f_ip(1-p)+(1-f_i)q(1-q))) \nonumber \\
\implies \operatorname{Var}^*(\mathcal{T}(i))=\frac{n^2\sum\limits_{j=1}^n(q\pi_j(1-q\pi_j))}{(\sum\limits_{j=1}^n\pi_j)^2(p-q)^2}\nonumber
\end{align}
\normalsize
\end{proof}

\section{Conclusion and way forward}
Sampling of nodes and its impact on accuracy of the trained models, statistical analysis of the data, and aspects of privacy have been at the epicentre of research in the areas of federated learning. The results in this paper enable us to have an unbiased estimate for the frequency of elements of a domain of values which are held by the users. We also get an insight on how the sampling affects the utility of the estimators and the accuracy of estimating the true distribution of the data. 

In Figure \ref{TV dist}, we observe that after a point, the TV distance doesn't decrease significantly compared to how much the communication cost increases, raising some interesting open questions: Should we go on till sampling every single node? Where should we stop? In fact, the first plot of Figure \ref{TV dist} shows sampling each node with probability 0.1 and sampling every single node do not engender a drastic difference in the TV distance. In particular, we would like to highlight some interesting open questions leading on from this work:

\begin{itemize}
\item[i)]\emph{Uniform sampling}: As we proposed an unbiased frequency estimator of the values which are sampled from the users with any arbitrary probability distribution, it would be an interesting area of analysis to first get an initial idea of the sampling distribution in the first round using $\mathcal{T}$, and then use that to our advantage to revise the sampling probabilities of each value inversely proportional to their frequencies so that we can ensure that a sample of the dataset we wish to derive doesn't over-represent a certain value and under-represent some others. In the context of FL, this can be a key area for ensuring a fair model which is not heavily influenced by the mode of the data, making the model more biased towards the majorities, which might not be the desirable outcome for certain tasks, e.g., facial recognition, text prediction, etc. It would be a challenging area to investigate how such a mechanism would perform in the aspect of the communication cost vs utility trade-off against the state-of-the-art differentially private FL techniques~\cite{Abadi2016}, especially for high dimensional data.

    \item[ii)]\emph{Shuffling}: Privacy amplification methods have been recently studied a lot involving the shuffle model. If we look to apply shuffling to the LDP data using DE as the local randomizer, that should mean we should have a high level of central differential privacy guarantee using a lower intensity of local noise using the recent advancements and studies for deriving the amplified formal central differential privacy guarantees using shuffling~\cite{balle2019privacy,feldman2020hiding,koskela2021tight,erlingsson2019amplification}. As the estimators we proposed, both $c_{\text{DE}}$ and $\mathcal{T}$, are functions function of the underlying LDP mechanism used -- in particular, the obfuscating probability distribution which is dependant on $\epsilon$ -- it is obvious that a higher value of $\epsilon$ will engender a better bound. The introduction of shuffling would guarantee that the privacy of the users would not be compromised, as we can tune the final level of central DP guarantee quite high for even a high value of the privacy parameter of DE, which is the local randomizer used in this process.
        
    Thus, it would be an interesting comparison to have between variance bounds of the estimated frequencies of the shuffle model with DE using our proposed estimates, and the variance of the observed data under the central Gaussian mechanism, which is essentially the maximum likelihood estimate of the original distribution of the data, under the same level of the privacy parameters. Depending on the behaviour, we could hypothesize on the requirement of the number of samples and the sampling probabilities that would ensure a tighter variance for our proposed estimates. 
        
    \item[iii)]\emph{Personalised sampling}: Another very interesting direction this work leads on to is to see if techniques like the Lagrange multiplies could be used to find the optimal sampling distribution ($\pi_1,\ldots,\pi_n$) that would minimize the variance of the estimator that we derived under the constraint that ($\pi_1,\ldots,\pi_n$) is a probability distribution. In other words, we would like to focus on the optimization problem where we wish to $\operatorname{Var}^*(\mathcal{T})(i)$ for every value $i\in\mathcal{X}$ such that $0\leq\pi_j\leq 1$ for every $j\in\{1,\ldots,n\}$ and $\sum\limits_{j=1}^n$. The problem would be straightforward if we wished to minimize $\operatorname{Var}^*(\mathcal{T})(i)$ for some fixed $i$, but becomes increasingly challenging when we are dealing with minimizing all the variances at an the minimum, under some multi-dimensional metric, giving us the optimal $(\pi_1,\ldots,\pi_n)$. This approach would enable us to find the optimal sampling probability that would give the minimum variance for our proposed unbiased estimators.
\end{itemize}

\bibliographystyle{IEEEtran}
\bibliography{references}

\begin{thebibliography}{10}
\providecommand{\url}[1]{#1}
\csname url@samestyle\endcsname
\providecommand{\newblock}{\relax}
\providecommand{\bibinfo}[2]{#2}
\providecommand{\BIBentrySTDinterwordspacing}{\spaceskip=0pt\relax}
\providecommand{\BIBentryALTinterwordstretchfactor}{4}
\providecommand{\BIBentryALTinterwordspacing}{\spaceskip=\fontdimen2\font plus
\BIBentryALTinterwordstretchfactor\fontdimen3\font minus
  \fontdimen4\font\relax}
\providecommand{\BIBforeignlanguage}[2]{{%
\expandafter\ifx\csname l@#1\endcsname\relax
\typeout{** WARNING: IEEEtran.bst: No hyphenation pattern has been}%
\typeout{** loaded for the language `#1'. Using the pattern for}%
\typeout{** the default language instead.}%
\else
\language=\csname l@#1\endcsname
\fi
#2}}
\providecommand{\BIBdecl}{\relax}
\BIBdecl

\bibitem{DworkDP1}
C.~Dwork, F.~McSherry, K.~Nissim, and A.~Smith, ``Calibrating noise to
  sensitivity in private data analysis,'' in \emph{Theory of Cryptography},
  S.~Halevi and T.~Rabin, Eds.\hskip 1em plus 0.5em minus 0.4em\relax Berlin,
  Heidelberg: Springer Berlin Heidelberg, 2006, pp. 265--284.

\bibitem{DworkDP2}
C.~Dwork, K.~Kenthapadi, F.~McSherry, I.~Mironov, and M.~Naor, ``Our data,
  ourselves: Privacy via distributed noise generation,'' in \emph{Advances in
  Cryptology - EUROCRYPT 2006}, S.~Vaudenay, Ed.\hskip 1em plus 0.5em minus
  0.4em\relax Berlin, Heidelberg: Springer Berlin Heidelberg, 2006, pp.
  486--503.

\bibitem{DuchiLDP}
J.~C. Duchi, M.~I. Jordan, and M.~J. Wainwright, ``Local privacy and
  statistical minimax rates,'' in \emph{2013 IEEE 54th Annual Symposium on
  Foundations of Computer Science}, 2013, pp. 429--438.

\bibitem{McMahanMRA16FL}
\BIBentryALTinterwordspacing
H.~B. McMahan, E.~Moore, D.~Ramage, and B.~A. y~Arcas, ``Federated learning of
  deep networks using model averaging,'' \emph{CoRR}, vol. abs/1602.05629,
  2016. [Online]. Available: \url{http://arxiv.org/abs/1602.05629}
\BIBentrySTDinterwordspacing

\bibitem{hard2019federated}
A.~Hard, K.~Rao, R.~Mathews, S.~Ramaswamy, F.~Beaufays, S.~Augenstein,
  H.~Eichner, C.~Kiddon, and D.~Ramage, ``Federated learning for mobile
  keyboard prediction,'' 2019.

\bibitem{yang2018applied}
T.~Yang, G.~Andrew, H.~Eichner, H.~Sun, W.~Li, N.~Kong, D.~Ramage, and
  F.~Beaufays, ``Applied federated learning: Improving google keyboard query
  suggestions,'' 2018.

\bibitem{chen2019federated}
M.~Chen, R.~Mathews, T.~Ouyang, and F.~Beaufays, ``Federated learning of
  out-of-vocabulary words,'' 2019.

\bibitem{ramaswamy2019federated}
S.~Ramaswamy, R.~Mathews, K.~Rao, and F.~Beaufays, ``Federated learning for
  emoji prediction in a mobile keyboard,'' 2019.

\bibitem{leroy2019federated}
D.~Leroy, A.~Coucke, T.~Lavril, T.~Gisselbrecht, and J.~Dureau, ``Federated
  learning for keyword spotting,'' 2019.

\bibitem{Abadi2016}
\BIBentryALTinterwordspacing
M.~Abadi, A.~Chu, I.~Goodfellow, H.~B. McMahan, I.~Mironov, K.~Talwar, and
  L.~Zhang, ``Deep learning with differential privacy,'' \emph{Proceedings of
  the 2016 ACM SIGSAC Conference on Computer and Communications Security}, Oct
  2016. [Online]. Available: \url{http://dx.doi.org/10.1145/2976749.2978318}
\BIBentrySTDinterwordspacing

\bibitem{dwork2010pan}
C.~Dwork, M.~Naor, T.~Pitassi, G.~N. Rothblum, and S.~Yekhanin, ``Pan-private
  streaming algorithms.'' in \emph{ics}, 2010, pp. 66--80.

\bibitem{rizk2020federated}
E.~Rizk, S.~Vlaski, and A.~H. Sayed, ``Federated learning under importance
  sampling,'' 2020.

\bibitem{RizkOpt21}
\BIBentryALTinterwordspacing
------, ``Optimal importance sampling for federated learning,'' 2020. [Online].
  Available: \url{https://arxiv.org/abs/2010.13600}
\BIBentrySTDinterwordspacing

\bibitem{WangPureProtocols}
\BIBentryALTinterwordspacing
T.~Wang, J.~Blocki, N.~Li, and S.~Jha, ``Locally differentially private
  protocols for frequency estimation,'' in \emph{26th {USENIX} Security
  Symposium ({USENIX} Security 17)}.\hskip 1em plus 0.5em minus 0.4em\relax
  Vancouver, BC: {USENIX} Association, Aug. 2017, pp. 729--745. [Online].
  Available:
  \url{https://www.usenix.org/conference/usenixsecurity17/technical-sessions/presentation/wang-tianhao}
\BIBentrySTDinterwordspacing

\bibitem{Cai2020}
L.~Cai, D.~Lin, J.~Zhang, and S.~Yu, ``Dynamic sample selection for federated
  learning with heterogeneous data in fog computing,'' in \emph{ICC 2020 - 2020
  IEEE International Conference on Communications (ICC)}, 2020, pp. 1--6.

\bibitem{Cormode21}
\BIBentryALTinterwordspacing
G.~Cormode, S.~Maddock, and C.~Maple, ``Frequency estimation under local
  differential privacy,'' \emph{Proc. VLDB Endow.}, vol.~14, no.~11, p.
  2046–2058, jul 2021. [Online]. Available:
  \url{https://doi.org/10.14778/3476249.3476261}
\BIBentrySTDinterwordspacing

\bibitem{kairouz2016discrete}
P.~Kairouz, K.~Bonawitz, and D.~Ramage, ``Discrete distribution estimation
  under local privacy,'' in \emph{International Conference on Machine
  Learning}.\hskip 1em plus 0.5em minus 0.4em\relax PMLR, 2016, pp. 2436--2444.

\bibitem{Wald}
A.~Wald, ``Small summaries for big data,''
  \url{https://en.wikipedia.org/wiki/Wald\%27s_equation}, accessed:
  12-Jul-2020.

\bibitem{balle2019privacy}
B.~Balle, J.~Bell, A.~Gasc{\'o}n, and K.~Nissim, ``The privacy blanket of the
  shuffle model,'' in \emph{Annual International Cryptology Conference}.\hskip
  1em plus 0.5em minus 0.4em\relax Springer, 2019, pp. 638--667.

\bibitem{feldman2020hiding}
V.~Feldman, A.~McMillan, and K.~Talwar, ``Hiding among the clones: A simple and
  nearly optimal analysis of privacy amplification by shuffling,'' \emph{arXiv
  preprint arXiv:2012.12803}, 2020.

\bibitem{koskela2021tight}
A.~Koskela, M.~A. Heikkil{\"a}, and A.~Honkela, ``Tight accounting in the
  shuffle model of differential privacy,'' \emph{arXiv preprint
  arXiv:2106.00477}, 2021.

\bibitem{erlingsson2019amplification}
{\'U}.~Erlingsson, V.~Feldman, I.~Mironov, A.~Raghunathan, K.~Talwar, and
  A.~Thakurta, ``Amplification by shuffling: From local to central differential
  privacy via anonymity,'' in \emph{Proceedings of the Thirtieth Annual
  ACM-SIAM Symposium on Discrete Algorithms}.\hskip 1em plus 0.5em minus
  0.4em\relax SIAM, 2019, pp. 2468--2479.

\end{thebibliography}

\end{multicols}

\end{document}